\documentclass[journal]{IEEEtran}
\usepackage{amsmath}
\usepackage{mathrsfs}
\usepackage{pifont}
\usepackage{bbding}
\usepackage{amsmath,epsfig}
\usepackage{mathbbold}
\usepackage{stmaryrd}
\usepackage{amssymb}
\usepackage{amsfonts}
\usepackage{epic}

\usepackage{amsfonts}
\usepackage{slashbox}
\usepackage{amsmath}
\usepackage{multirow}
\usepackage{graphicx}
\usepackage{amsfonts}
\usepackage{amsmath,epsfig}
\usepackage{mathbbold}
\usepackage{stfloats}
\usepackage{amssymb}
\usepackage{url}
\usepackage{bm}
\usepackage{cite}
\usepackage{amsmath}
\usepackage{amssymb}
\usepackage{latexsym}
\usepackage{multirow}
\usepackage{epsfig}
\usepackage{graphics}
\usepackage{mathrsfs}
\usepackage{algorithmic}
\usepackage{algorithm}
\usepackage{subfigure}
\usepackage{slashbox}
\usepackage[all]{xy}

\usepackage{pifont}
\usepackage{bbding}
\usepackage{stmaryrd}
\usepackage{amssymb}
\usepackage{amsfonts}
\usepackage{epic}
\usepackage{stfloats}
\usepackage{epstopdf}
\usepackage{bm}
\usepackage{xcolor}
\usepackage{mathrsfs}
\usepackage{pifont}
\usepackage{bbding}
\usepackage{amsmath,epsfig}
\usepackage{mathbbold}
\usepackage{stmaryrd}
\usepackage{amsmath}
\usepackage{amssymb}
\usepackage{amsthm}
\usepackage{amsfonts}
\usepackage{epic}
\usepackage{graphicx}
\usepackage{subfigure}
\usepackage{enumerate}
\usepackage{latexsym}
\usepackage{epstopdf}
\usepackage{multirow}
\usepackage{stfloats}
\usepackage{bm}
\usepackage{booktabs}
\usepackage{color}
\usepackage{cases}
\usepackage{array}
\usepackage{makecell}
\usepackage{times}

\usepackage{diagbox}
\usepackage{slashbox}
\usepackage{booktabs}

\newtheorem{proposition}{Proposition}

\newcommand{\Q}{\bm Q}

\newcommand{\w}{\bm w}

\newcommand{\V}{\bm V}

\newcommand{\Ss}{\bm S}

\newcommand{\q}{\bm q}

\newcommand{\vvv}{\bm v}

\newcommand{\ttheta}{\mathbf \Theta}

\graphicspath{{./figs/}}

\begin{document}
\title{   {IRS-Assisted Wireless Powered NOMA: Do We Really Need Different Phase Shifts in DL and UL?}}
\author{\IEEEauthorblockN{Qingqing Wu, \emph{Member, IEEE}, Xiaobo Zhou, and Robert Schober,  \emph{Fellow, IEEE}
\thanks{    Q. Wu is  with the State Key Laboratory of Internet of Things for Smart City, University of Macau, China 999078 (email: qingqingwu@um.edu.mo).    }
   } }
\maketitle
%
\begin{abstract}
Intelligent reflecting surface (IRS) is a promising technology to improve the performance of wireless powered communication networks (WPCNs)  due to its capability to reconfigure signal propagation environments via smart reflection. In particular, the high passive beamforming gain promised by IRS can significantly enhance the efficiency of both downlink wireless power transfer (DL WPT) and uplink wireless information transmission (UL WIT) in WPCNs. Although adopting different IRS phase shifts for DL WPT and UL WIT, i.e., {\it dynamic IRS  beamforming}, is in principle possible but incurs additional signaling overhead and computational complexity, it is an open problem whether it  is actually beneficial. To answer this question, we consider an IRS-assisted WPCN where multiple devices employ a hybrid access point (HAP) to first harvest energy and then transmit information using  non-orthogonal multiple access (NOMA).  Specifically, we aim to maximize the sum throughput of all devices by jointly optimizing the IRS phase shifts and the resource allocation. To this end, we first prove that dynamic IRS beamforming is not needed for the considered system, which helps reduce the number of IRS phase shifts to be optimized.  Then, we propose both joint and alternating optimization based algorithms to solve the resulting problem. {Simulation results demonstrate the effectiveness of  our proposed designs over benchmark schemes and also provide useful insights into the importance of IRS for realizing spectrum and energy efficient WPCNs.}

\end{abstract}
\begin{IEEEkeywords}
IRS, wireless powered networks, NOMA, dynamic  beamforming, time allocation.
\end{IEEEkeywords}
\section{introduction}


Intelligent reflecting surface (IRS) has been recently proposed as a cost-effective technology to improve the spectral efficiency and energy efficiency of future wireless networks \cite{JR:wu2019IRSmaga}. Specifically, by smartly adjusting the phase shifts of a large number of reflecting elements, IRS can reconfigure the wireless propagation environment to achieve different design objectives, such as signal focusing and interference suppression. {There has been an upsurge of interest in investigating joint active and passive beamforming  for various system setups \cite{rajatheva2020white,xu2020resource,guan2019intelligent,zou2020wireless,ding2020simple,JR:wu2019discreteIRS,JR:wu2018IRS,huang2018achievable,huang2019reconfigurable,di2020hybrid}.
 In particular,  the fundamental \emph{squared power gain} of IRS was  unveiled  in  \cite{JR:wu2018IRS}.} 
{In \cite{huang2019reconfigurable} and \cite{di2020hybrid},  energy efficiency and achievable rate maximization problems were studied by considering continuous and discrete IRS phase shifts, respectively.}


\begin{figure}[!t]
\centering
\includegraphics[width=2.7in]{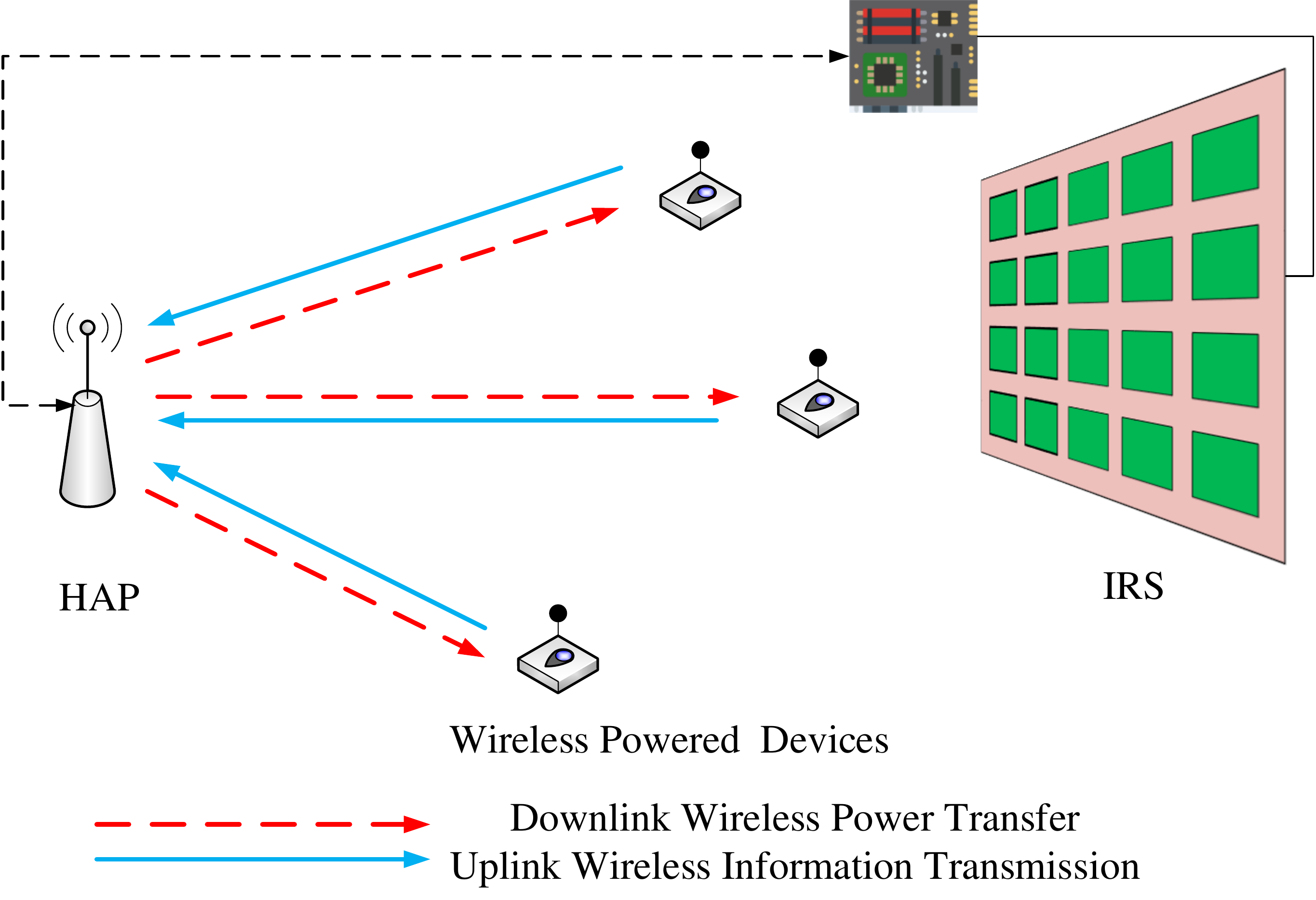}
\caption{An IRS-assisted WPCN employing NOMA for UL WIT.}\label{system:model} \vspace{-0.5cm}
\end{figure}

 While the above works  focused on applying IRS for improving  wireless information transmission (WIT),  there is a growing interest in exploiting the high beamforming gain of IRS to enhance the efficiency of wireless power transfer (WPT) to Internet-of-Things (IoT) devices.   One line of research targets IRS passive beamforming for simultaneous wireless information and power transfer (SWIPT), where information and energy receivers are concurrently served \cite{wu2019weighted,wu2019jointSWIPT,tang2020joint,liu2020energy}.  The other line of research focuses on IRS-assisted wireless powered communication networks (WPCNs), where self-sustainable devices first harvest energy in the downlink (DL) and then transmit information signals in the uplink (UL)
\cite{lyubin2021IRS,zheng2020joint}, based on the well-known ``harvest and then transmit''  protocol illustrated in Fig. \ref{system:model}. The WPCN sum throughput maximization problem was studied in \cite{lyubin2021IRS}  for UL WIT with time-division multiple access (TDMA). To improve the WIT efficiency, space-division multiple access (SDMA) was employed for UL WIT in  \cite{zheng2020joint} by jointly optimizing the IRS phase shifts and the transmit powers. 

Non-orthogonal multiple access (NOMA) is practically appealing for UL WIT in WPCNs due to its capability to improve spectrum efficiency and user fairness by allowing multiple users to simultaneously access the same spectrum. To the best of the authors' knowledge, optimization of IRS-assisted WPCNs employing NOMA for UL WIT has not been studied  in the literature yet.  Furthermore,  since DL WPT and  UL WIT occur in different time slots and have different objectives, it is usually believed that exploiting different IRS phase-shift vectors in the two phases, which is referred to as \emph{dynamic IRS beamforming},  may improve system performance. As such, all existing works on IRS-assisted WPCNs (e.g., \cite{lyubin2021IRS,zheng2020joint}) naturally assume dynamic passive beamforming in their problem formulations. However,  it remains an open problem whether dynamic IRS beamforming is actually beneficial for maximizing the throughput of wireless powered IoT networks employing NOMA.

Motivated by the above considerations, we study an IRS-assisted WPCN where an IRS is deployed to assist the DL WPT and UL WIT between a hybrid access point (HAP) and multiple devices.
For this setup, we aim to maximize the sum throughput of all devices  by jointly optimizing the resource allocation and IRS phase shifts (i.e., passive beamforming). For the first time, we unveil that adopting different IRS phase shifts for DL WPT and UL WIT, i.e.,  dynamic IRS beamforming, is not beneficial for the considered system. Since the algorithmic computations are typically executed at the HAP, which then sends the optimized phase shifts to the IRS controller, this result not only reduces the number of optimization variables, but also lowers the feedback signalling overhead, especially when the IRS is practically large. Exploiting the insight gained, we further propose both joint and alternating optimization based algorithms to solve the resulting problem. Numerical results  demonstrate the significant performance gains achieved by the proposed algorithms compared to benchmark schemes and reveal that  integrating IRS into WPCNs not only improves the system throughput but also reduces the system energy consumption.

\emph{Notations:}  For a vector $\bm{x}$,  $[\bm{x}]_n$ denotes its $n$-th entry and $\bm{x}^T$ and $\bm{x}^H$ denote its transpose and Hermitian transpose, respectively. $\mathcal{O}(\cdot)$ denotes the computational complexity order. $ \mathrm{Re}\{\cdot\}$ denotes the real part of a complex number.  ${\rm{tr}}(\Ss)$ denotes the  trace of matrix $\Ss$. $\arg(\bm{x})$ denotes the phase vector of $\bm{x}$.

\section{System Model and Problem Formulation}
\subsection{System Model}

As shown in Fig. 1, we consider an IRS-aided WPCN, which comprises an HAP, an IRS, and $K$ wireless-powered devices. The HAP and the devices are all equipped with a single antenna and the IRS consists of $N$ reflecting elements.
To ease practical implementation, the HAP and all devices are assumed to operate over the same frequency band, with the total available transmission time denoted by $T_{\max}$.  In addition, the quasi-static flat-fading channel model is adopted which means that the channel coefficients remain constant during $T_{\max}$. As such,  UL/DL channel reciprocity holds for all channels, which allows channel state information (CSI) acquisition for the DL based on UL training.  The WPCN adopts the typical ``harvest and then transmit'' protocol  where the devices first harvest energy from the DL signal emitted by the HAP and then use the available energy to transmit information to the HAP in the UL.    
To be able to characterize the maximum achievable performance, it is assumed that the CSI of all channels is perfectly known  at the HAP.  The equivalent baseband  channels  from the HAP to the IRS, from the IRS to device $k$, and from the HAP to device $k$ are denoted by $\bm{g}\in \mathbb{C}^{N\times 1}$, $\bm{h}^H_{r,k}\in \mathbb{C}^{1\times N}$, and ${h}^H_{d,k}\in \mathbb{C}$, respectively, where $k = 1, \cdots,K$.




 %

During DL WPT, the HAP broadcasts an energy signal with constant transmit power $P_{\rm A}$ during time $\tau_0$.
 The energy harvested from the noise is assumed to be negligible as in \cite{ju14_throughput},  since  the noise power is much smaller than the  power received from the HAP. Let $\ttheta_0 = \text{diag} ( e^{j\theta_1}, \cdots,  e^{j\theta_N})$ denote the reflection phase-shift  matrix of the IRS for DL WPT where $\theta_n\in [0, 2\pi), \forall n$.
 Thus, the amount of harvested energy  at device $k$ can be expressed as
\begin{align}\label{eq3}
E^h_k&=\eta_kP_{\rm A}|h^H_{d,k} +  \bm{h}^H_{r,k}\ttheta_0 \bm{g}|^2\tau_0  \\ \nonumber
&=\eta_kP_{\rm A}|h^H_{d,k} +   \bm{q}_k^H \vvv_0|^2\tau_0,
\end{align}
where $\eta_k \in (0,1]$ is the  energy conversion efficiency of device $k$, $\q_k= \bm{h}^H_{r,k}   \text{diag}(\bm{g})$, and $\vvv_0 = [e^{j\theta_1}, \cdots, e^{j\theta_N}]^T$.
For UL WIT, NOMA is adopted where all devices transmit their respective information signals to the HAP simultaneously for a duration of ${\tau}_{\rm 1}$  with transmit powers $p_k$, $k=1,\cdots,K$.  The HAP applies successive interference cancellation (SIC)  to eliminate multiuser interference.   
 Specifically, for detecting the message of the $k$-th device, the HAP first decodes the message of the $i$-th device, $\forall i<k$, and then removes this message from the received signal, in the order of $i=1, 2,...,k-1$. The signal received from the $i$-th user, $\forall i>k$, is treated as noise. Hence, the achievable throughput of device $k$  in bits/Hz   can be expressed as
\begin{align}\label{eq6}
r_k={\tau}_{\rm 1}\log_2\left(1+\frac{p_k |h^H_{d,k} + \q^H_k \vvv_1|^2 }{\sum_{i=k+1}^{K} p_i |h^H_{d,i} +   \q^H_i \vvv_1|^2+ \sigma^2}\right),
\end{align}
where $\sigma^2$  is the additive white Gaussian noise power at the HAP and $\vvv_1 = [ e^{j\varphi_1}, \cdots,  e^{j\varphi_N}]^T$  denotes the IRS phase shift vector for UL WIT.
  Then, the sum throughput is given by \cite{zeng2019spectral}
\begin{align}\label{eq6}
 \!\!\!\! R_{\rm sum}\!= \!\sum_{k=1}^{K}r_k \! =\! {\tau}_{\rm 1} \log_2\left( 1+\sum_{k=1}^{K}\frac{p_k|h^H_{d,k} +   \q^H_k \vvv_1 |^2}{\sigma^2}   \right).
\end{align}


\subsection{Problem Formulation}
Our objective is to maximize the  sum throughput of the considered system by jointly optimizing the IRS phase shifts, the time allocation, and the transmit powers. This leads to the following optimization problem
 \begin{subequations} \label{probm20}
\begin{align}
\text{(P1)}:  \mathop {\mathrm{max} }\limits_{   \overset{ \tau_{0}, {\tau}_{\rm 1},\{p_{k}\},  } { \vvv_0, \vvv_1 }  }   &{\tau}_{\rm 1}\log_2\left(1+\sum_{k=1}^{K}\frac{p_k|h^H_{d,k} +   \q^H_k \vvv_1 |^2}{\sigma^2}\right)\\
\mathrm{s.t.} ~~&  {p_k}{\tau}_{\rm 1}\leq \eta_kP_{\rm A} |h^H_{d,k} +   \q^H_k \vvv_0 |^2 \tau_0, ~ \forall k,  \label{eq201}\\
&\tau_{0}+{\tau}_{\rm 1}\leq T_{\mathop{\max}}, \label{eq202}\\
&\tau_{0}\geq0, ~ {\tau}_{\rm 1}\geq  0, ~p_k\geq  0, ~\forall k, \label{eq203} \\
& |[\vvv_0]_n|=1,  n=1,\cdots, N, \label{eq:modulus1} \\
& |[\vvv_1]_n|=1,  n=1,\cdots, N. \label{eq:modulus2}
\end{align}
\end{subequations}
In (P1), (\ref{eq201}) and (\ref{eq202}) represent the energy causality and total time constraints, respectively, (\ref{eq203}) are non-negativity constraints, and  \eqref{eq:modulus1} and \eqref{eq:modulus2} are unit-modulus constraints for the phase shifts employed for DL WPT and UL WIT, respectively. Intuitively, since the DL and UL transmissions have different objectives, i.e., WPT and WIT, adopting different phase-shift vectors, i.e., $\vvv_0$ and $\vvv_1$, is expected to be beneficial for maximizing the system sum throughput. However, using different IRS phase-shift vectors also increases the feedback signalling overhead to the IRS and the computational complexity at the HAP due to the large number of optimization variables.  Furthermore, (P1) is a non-convex optimization problem and difficult to solve optimally in general due to coupled optimization variables in (4a) and (4b) as well as the non-convex unit-nodulus constraints in (4e) and (4f).

\section{Proposed Solution}
In this section, we first answer the question whether the optimal solution of (P1) requires  dynamic IRS beamforming. Then, we propose two efficient algorithms to solve the resulting optimization problem.

\vspace{-0.1cm}
\subsection{Do We  Need Different Phase Shifts in DL and UL?}
\begin{proposition}
For (P1),  $\vvv^{\star}_0=\vvv^{\star}_1$ holds, where $\vvv^{\star}_0$ and $\vvv^{\star}_1$ are the optimal phase-shift vectors for DL WPT and UL WIT, respectively.
\end{proposition}
\begin{proof}
First, it can be easily shown that constraint \eqref{eq201} is met with equality for the optimal solution since  otherwise $p_k$ can be always increased to improve the objective value until  \eqref{eq201} becomes active. Then, substituting  \eqref{eq201} into the objective function eliminates  $\{p_k\}$.   As such, for any given $\tau_0$ and $\tau_1$, (P1) is equivalent to
 \begin{subequations} \label{sub:probm20}
\begin{align}
 \mathop {\mathrm{max} }\limits_{{ \vvv_0, \vvv_1 } } ~~&   \sum_{k=1}^{K}  \alpha_k    | h^H_{d,k} +   \q^H_k \vvv_0 |^2| h^H_{d,k} +  \q^H_k \vvv_1 |^2    \\
\mathrm{s.t.} ~~
& \text{(4e), (4f),}
\end{align}
\end{subequations}
where  $\alpha_k=\frac{\tau_0 P_{\rm A}  \eta_k } {  {\tau}_{\rm 1}\sigma^2  }$.
Denote by $\w$ the vector maximizing $ \sum_{k=1}^{K}\alpha_k| h^H_{d,k} + \q^H_k \vvv |^4$ subject to constraints $|[\vvv]_n|=1, \forall n$. For the objective function in (5a), we can establish  the following inequalities
{\begin{align} 
\!\!\!\!  &\sum_{k=1}^{K}    \left(  \sqrt{\alpha_k  }  | h^H_{d,k} + \q^H_k \vvv_0 |^2 \right)   \left( \sqrt{\alpha_k  } | h^H_{d,k} +\q^H_k \vvv_1 |^2 \right)  \\
\!\!\!\! \overset{(a)}{\leq} &\sqrt{ \left(\sum_{k=1}^{K} \alpha_k| h^H_{d,k} +\q^H_k \vvv_0 |^4\right)\left( \sum_{k=1}^{K}\alpha_k | h^H_{d,k} + \q^H_k \vvv_1 |^4\right) }  \nonumber \\
 \!\!\!\! \overset{(b)}{\leq}& \sqrt{ \left(\sum_{k=1}^{K}\alpha_k | h^H_{d,k} +\q^H_k \w |^4\right)^2 }= { \sum_{k=1}^{K} \alpha_k| h^H_{d,k} + \q^H_k \w|^4},
\end{align}}where $(a)$ is due to the Cauchy-Schwarz inequality and $(b)$ holds since $\w$ maximizes $ \sum_{k=1}^{K}\alpha_k| h^H_{d,k} + \q^H_k \vvv |^4$  and the equality holds when  $\vvv^{\star}_0=\vvv^{\star}_1=\w$. 
\end{proof}

Proposition 1  explicitly shows that dynamic IRS beamforming  is not needed for UL WPT and DL WIT and using constant passive beamforming is sufficient to maximize the sum throughput of the considered system. As such, if the HAP is in charge of computing the IRS phase shifts, it only needs to feed back $N$ phase-shift values (i.e., $\vvv_0$) to the IRS, rather than $2N$ (i.e., $\vvv_0$ and $\vvv_1$), which reduces the signalling overhead and the associated delay, especially for practically large $N$. Furthermore, exploiting  Proposition 1, we only need to solve the following problem, which involves a smaller number of optimization variables (i.e., IRS phase shifts) 
 \begin{subequations} \label{probm20}
\begin{align}
 \mathop {\mathrm{max} }\limits_{{\tau_{0}, {\tau}_{\rm 1},  \vvv_0 } } &{\tau}_{\rm 1}\log_2\left(1+\sum_{k=1}^{K}\frac{\tau_0 P_{\rm A}  \eta_k|h^H_{d,k} +     \q^H_k \vvv_0  |^4 }{{\tau}_{\rm 1}\sigma^2}\right)   \label{eqNdy:obj}\\
\mathrm{s.t.} ~~& \tau_{0}+{\tau}_{\rm 1}\leq T_{\mathop{\max}}, \label{eqNdy202}\\
&\tau_{0}\geq0, ~ {\tau}_{\rm 1}\geq  0, \label{eqNdy203} \\
& |[\vvv_0]_n|=1,  n=1,\cdots, N.  \label{eq:Ndy:modulus1}
\end{align}
\end{subequations}
Although simpler, problem \eqref{probm20} is still a non-convex optimization problem.

\subsection{Proposed Joint Optimization Algorithm}
To deal with the non-convex objective function \eqref{eqNdy:obj}, we introduce a slack variable $S$ and reformulate  problem  \eqref{probm20} as follows
 \begin{subequations} \label{probm:JO:20}
\begin{align}
 \mathop {\mathrm{max} }\limits_{{\tau_{0},{\tau}_{\rm 1}, \vvv_0 } } &{\tau}_{\rm 1}\log_2\left(1+\frac{S}{{\tau}_{\rm 1}}\right)\\
\mathrm{s.t.} ~~& S \leq  \sum_{k=1}^{K}\frac{\tau_0 P_{\rm A}  \eta_k|h^H_{d,k} +   \q^H_k \vvv_0  |^4 }{\sigma^2},  \label{eqJO201}\\
~~& \tau_{0}+{\tau}_{\rm 1}\leq T_{\mathop{\max}}, \label{eqJO202}\\
&\tau_{0}\geq0, ~{\tau}_{\rm 1}\geq  0,\label{eqJO203} \\
& |[\vvv_0]_n|=1,  n=1,\cdots, N. \label{eq:JO:modulus1}
\end{align}
\end{subequations}
{Note that for the optimal solution of problem \eqref{probm:JO:20}, constraint  \eqref{eqJO201} is met with equality, since otherwise we can always  increase the objective value by increasing $S$ until  \eqref{eqJO201} becomes active.}
Let $ | h^H_{d,k} +  \q^H_k \vvv_0  | = |   {\bar \q}^H_k \bar \vvv_0   | $, where ${\bar\vvv}_0 = [ \vvv^H_0 \: 1]^H$ and $ {\bar \q}^H_k  =  [  {\q}^H_k  \: h^H_{d,k}] $. Define $\Q_k={\bar \q}_k{\bar \q}^H_k$ and $\bm{V}_0=\bm{\bar{v}}_0\bm{\bar{v}}_0^H$ which needs to satisfy  $\bm{V}_0\succeq \bm{0}$ and ${\rm{rank}}(\bm{V}_0)=1$.
Then, \eqref{eqJO201}  can be expressed as
\begin{align} \label{eq:AO:slack}
S\leq   \sum_{k=1}^{K}\frac{\eta_k  P_{\rm A} \tau_0  [{\rm{Tr}}(\V_0\Q_k)]^2 }{\sigma^2}.
\end{align}

The key observation is that  although $\tau_0 [{\rm{Tr}}(\V_0\Q_k)]^2= [{\rm{Tr}}(\V_0\Q_k)]^2/\frac{1}{\tau_0}$  in \eqref{eq:AO:slack}  is not  jointly convex with respect to $\V_0$ and $\tau_0$, it is jointly convex with respect to ${\rm{Tr}}(\V_0\Q_k)$ and $\frac{1}{\tau_0}$. Recall that any convex function is globally lower-bounded by its first-order Taylor expansion at any point. This thus motivates us to apply the successive convex approximation (SCA) technique for solving problem \eqref{probm:JO:20}. Therefore, with given local point  $\hat \V_0$ and $\hat \tau_0$, we obtain the following lower bound
\begin{align} 
\!\!& \frac{ [{\rm{Tr}}(\V_0\Q_k)]^2}{{1}/{\tau_0}} \geq   \frac{ [{\rm{Tr}}(\hat \V_0\Q_k)]^2}{{1}/{\hat \tau_0}}  - \frac{ [{\rm{Tr}}(\hat \V_0\Q_k)]^2}{{1}/{\hat \tau_0}} \left(\frac{1}{\tau_0} - \frac{1}{\hat \tau_0} \right)
 \nonumber  \\
\!\!& + \frac{2 {\rm{Tr}}(\hat \V_0\Q_k)}{{1}/{\hat \tau_0}} \left(   {\rm{Tr}}( \V_0\Q_k) -  {\rm{Tr}}(\hat \V_0\Q_k) \! \right) \!    \triangleq \!  \mathcal{G}(\V_0,\tau_0).
\end{align}
With the lower bound in (11), problem \eqref{probm:JO:20} is approximated as the following problem
 \begin{subequations}\label{probm:SDR:20}
\begin{align}
 \mathop {\mathrm{max} }\limits_{{\tau_{0}, {\tau}_{\rm 1}, \V_0, S } } &{\tau}_{\rm 1}\log_2\left(1+\frac{S}{{\tau}_{\rm 1}}\right)\\
\mathrm{s.t.} ~~& S \leq  \sum_{k=1}^{K}\frac{P_{\rm A}  \eta_k }{\sigma^2}  \mathcal{G}(\V_0,\tau_0), \\
~~& \tau_{0}+ {\tau}_{\rm 1}\leq T_{\mathop{\max}}, \label{eqSDR202}\\
&\tau_{0}\geq0, ~{\tau}_{\rm 1}\geq  0,\label{eqSDR203} \\
~~~~& [\bm{V}_0]_{n,n} = 1,  n=1,\cdots, N+1, \label{P6:SDR:C9}  \\
~~~~&{\rm{rank}}(\bm{V}_0)=1, \bm{V}_0 \succeq 0.  \label{P6:SDR:C10}
\end{align}
\end{subequations}
Note that  by relaxing the rank-one constraint in \eqref{P6:SDR:C10}, problem \eqref{probm:SDR:20} becomes a convex semidefinite program (SDP)  and  we can successively solve it by using standard convex optimization solvers such as CVX, until convergence is achieved. After convergence, Gaussian randomization can be applied to obtain a high-quality solution. {Alternatively, instead of relaxing the rank-one constraint, we can further transform it into an equivalent constraint based on the largest singular value (see Section IV-A for details) and then solve the resulting problem using the penalty method \cite{wu2019jointSWIPT}.}  The computational complexity of this algorithm lies in solving the SDP and is given by $\mathcal{O}( I_{JO}N^{6.5}  )$ where $I_{JO}$ denotes the number of iterations required for convergence.  Since all optimization variables are optimized simultaneously, the joint optimization algorithm serves as a benchmark scheme for other schemes having lower complexities.
\vspace{-0.3cm}
\subsection{Proposed Alternating Optimization Algorithm}
Next, we propose an efficient alternating optimization algorithm where the phase shifts and time allocation are alternately optimized until convergence is achieved.  {The key advantage of the algorithm is that it admits a (semi) closed-form solution in each iteration, which thus avoids  the computational complexity incurred by using SDP solvers in Section III-B.}

First, for any fixed $ \vvv_0$, it can be shown that problem \eqref{probm20} is simplified to a convex optimization problem where the optimal time allocation can be obtained by using the Lagrange duality method as in \cite{ju14_throughput} where the details are omitted here for brevity. Second, for any fixed $\tau_0$ and $\tau_1$, problem \eqref{probm20} is reduced to
 \begin{subequations} \label{probm:AO:20}
\begin{align}
 \mathop {\mathrm{max} }\limits_{{ \vvv_0 } } ~~&   \sum_{k=1}^{K}  \alpha_k|h^H_{d,k} + \q^H_k \vvv_0 |^4 =    \sum_{k=1}^{K}  \alpha_k   (\bar \vvv_0^H \Q_k \bar \vvv_0)^2  \label{eq:AO:obj} \\
\mathrm{s.t.} ~~  
& |[\vvv_0]_n|=1,  n=1,\cdots, N, \label{eq:AO:modulus1}
\end{align}
\end{subequations}
where $\alpha_k=\frac{\tau_0 P_{\rm A}  \eta_k } {  {\tau}_{\rm 1}\sigma^2  }$ as in \eqref{sub:probm20}. Although maximizing a convex function does not lead to a convex optimization problem, the convexity of \eqref{eq:AO:obj}  allows us to apply the SCA technique for solving problem \eqref{probm:AO:20}. Specifically, for a given local point $ \hat \vvv_0$,  the  first-order Taylor expansion based lower bound for the $k$-th term in  \eqref{eq:AO:obj} can be expressed as  
\begin{align}\label{eqAO}
(\bar \vvv_0^H \Q_k \bar \vvv_0)^2 \geq &~2\hat \vvv_0^H \Q_k  \hat \vvv_0 \left(     2 \mathrm{Re}\{  \hat \vvv^H_0 \Q_k  \bar \vvv_0 \}  -  2  \hat \vvv^H_0 \Q_k  \hat \vvv_0    \right) \nonumber \\
&~+ (\hat \vvv_0^H \Q_k  \hat \vvv_0)^2   \nonumber \\
=& ~4 \mathrm{Re}\{ C_k \hat \vvv_0^H \Q_k   \bar\vvv_0 \} - 3C^2_k,
\end{align}
where $C_k=\hat \vvv_0^H \Q_k  \hat \vvv_0$.  Based on \eqref{eqAO},   a lower bound for \eqref{eq:AO:obj} is given by 
\begin{align}\label{SCA:obj}
  4 \mathrm{Re} \left\{  \left(\sum_{k=1}^{K}  \alpha_kC_k  \hat \vvv_0^H \Q_k   \right)   \bar\vvv_0 \right\} - 3  \sum_{k=1}^{K}\alpha_kC^2_k.
\end{align}
Based on \eqref{SCA:obj}, it is not difficult to show that the optimal solution satisfying \eqref{eq:AO:modulus1} is given by $ \bar\vvv_0=e^{j\arg( {\bm\beta})}$ where ${\bm \beta}= (   \sum_{k=1}^{K}  \alpha_kC_k  \hat \vvv_0^H \Q_k    )^H$.

Note that  the objective value of problem \eqref{probm20} is non-decreasing by alternately optimizing the time allocation and phase shifts, and also upper-bounded by a finite value. Thus, the proposed algorithm is guaranteed to converge.
The complexity of this algorithm mainly lies in the calculation of the phase shifts for problem \eqref{probm:AO:20} and thus is given by $\mathcal{O}( I_{AO}N )$, where $I_{AO}$ is the number of iterations required for convergence.





\section{Numerical Results }
This section presents simulation results to demonstrate the effectiveness of the proposed solutions and provide useful insights for IRS-aided WPCN design.  
 The HAP and IRS are respectively located at $(0,0,0)$ meter (m) and $(10,0, 3)$ m, as shown in Fig. \ref{simulation:setup}. The user devices are randomly and uniformly distributed within a radius of $1.5$ m centered at $(10,0,0)$ m. The pathloss exponents of both the HAP-IRS and IRS-device channels are set to $2.2$, while those of the HAP-device channels are set to $2.8$. {Furthermore, Rayleigh fading is adopted as the small-scale fading for all channels.} The signal attenuation at a reference distance of $1$ m is set as $30$ dB.  Other system parameters are set as follows: $\eta_k=0.8,~\forall k$, $\sigma^2=-85$ dBm,   $T_{\max}=1$ s, and $P_{\rm A}=40$ dBm.

\begin{figure}[!t]
\centering
\includegraphics[width=2.4in]{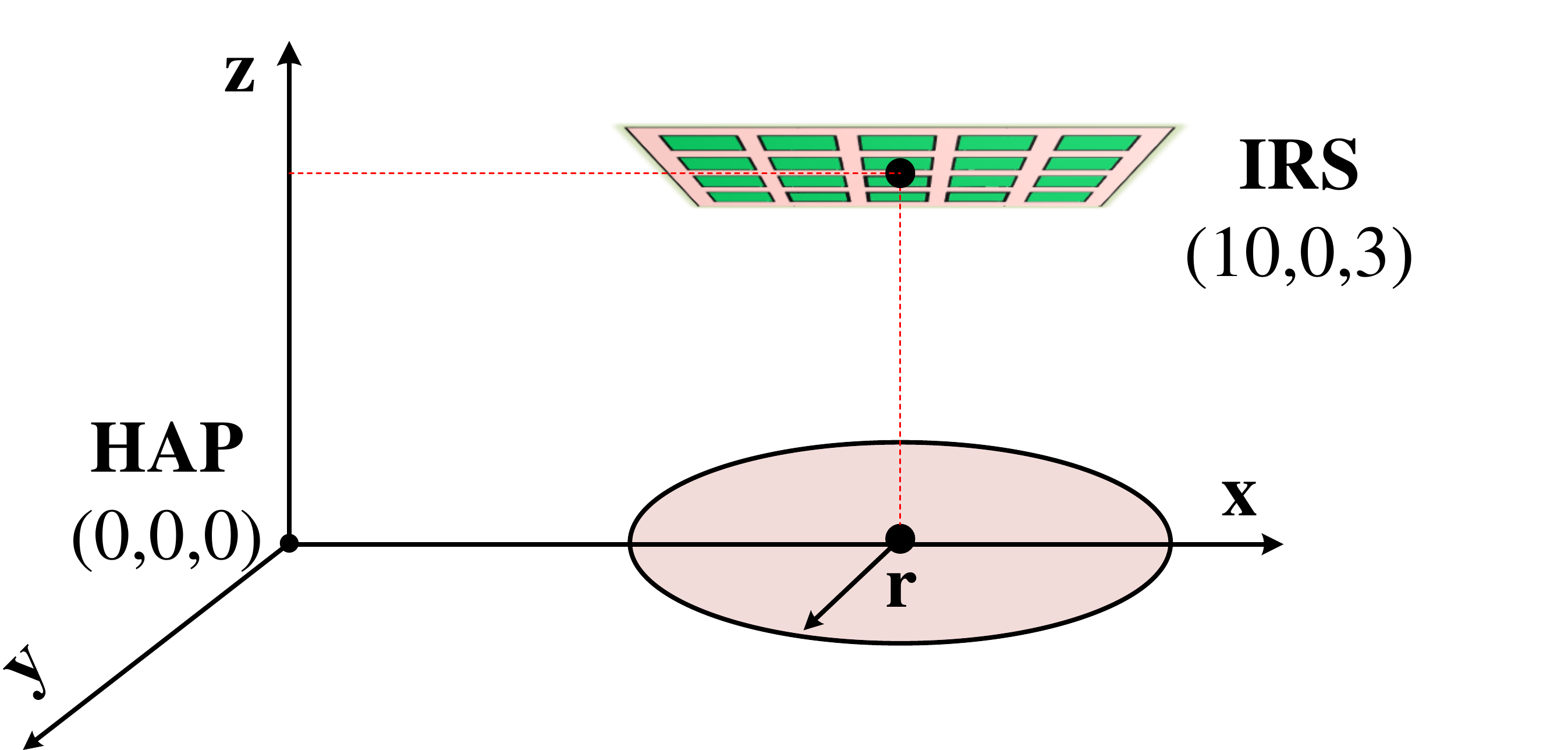}  
\caption{Simulation setup. } \vspace{-0.5cm}\label{simulation:setup}
\end{figure}

\begin{figure*}[!t]\vspace{-0.0cm}
\centering
\subfigure[Performance comparison.]  {\includegraphics[width=2.35in, height=1.8in]{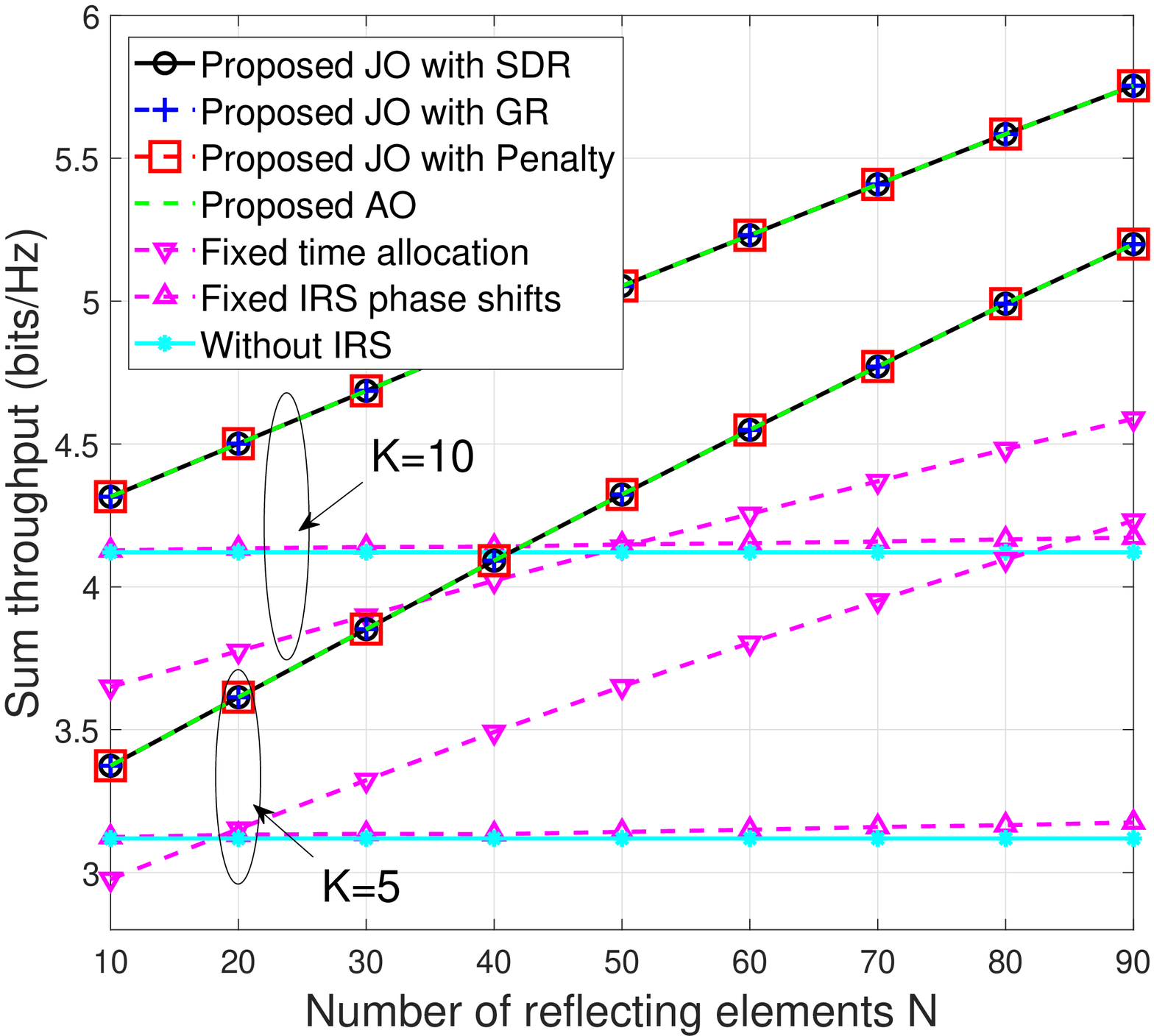}\label{N:versus:rate}}
\subfigure[Impact of $N$ on DL WPT duration.]{\includegraphics[width=2.35in, height=1.8in]{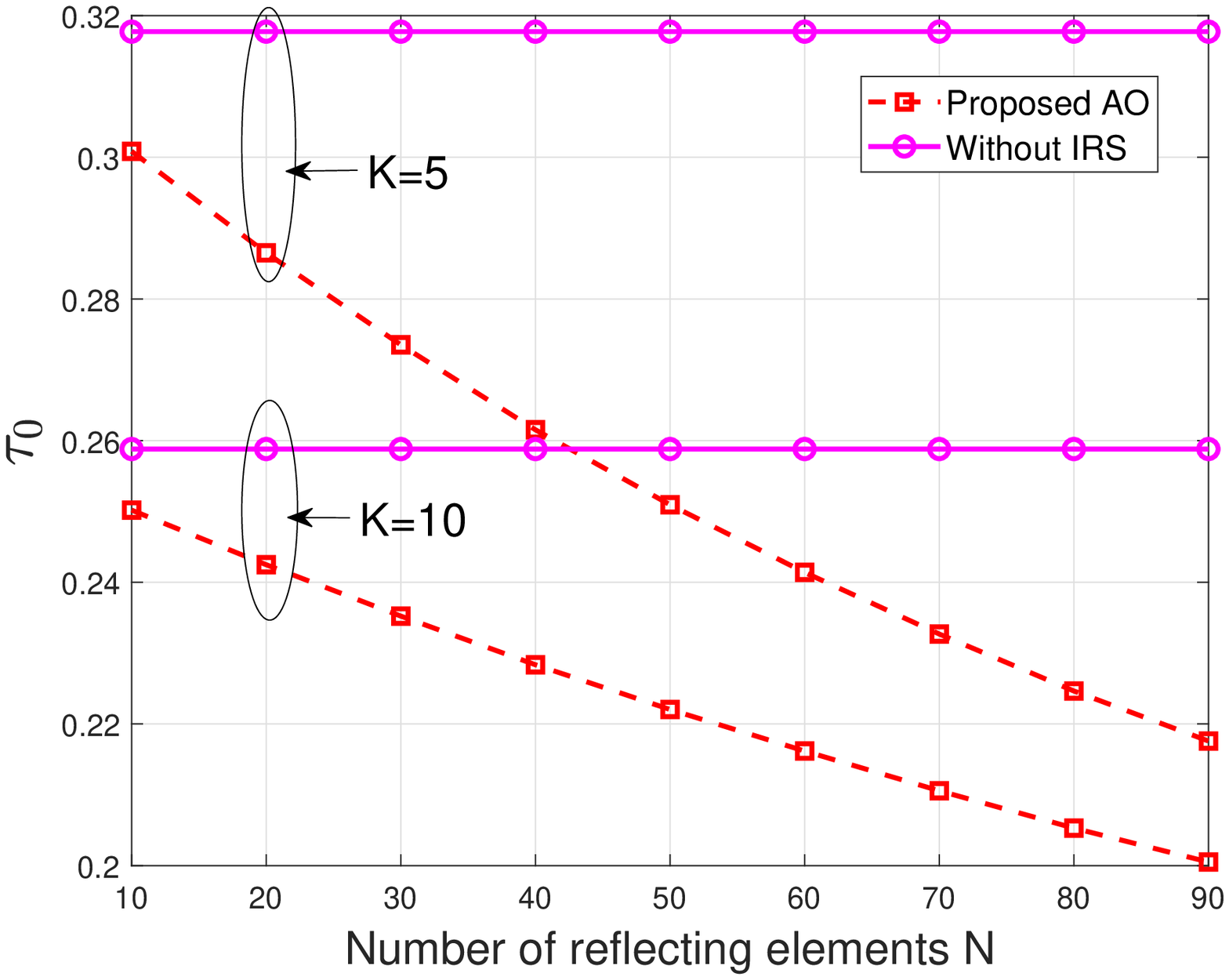}\label{N:versus:tau0}}
\subfigure[Impact of $N$ on user harvested energy.]{\includegraphics[width=2.35in, height=1.8in]{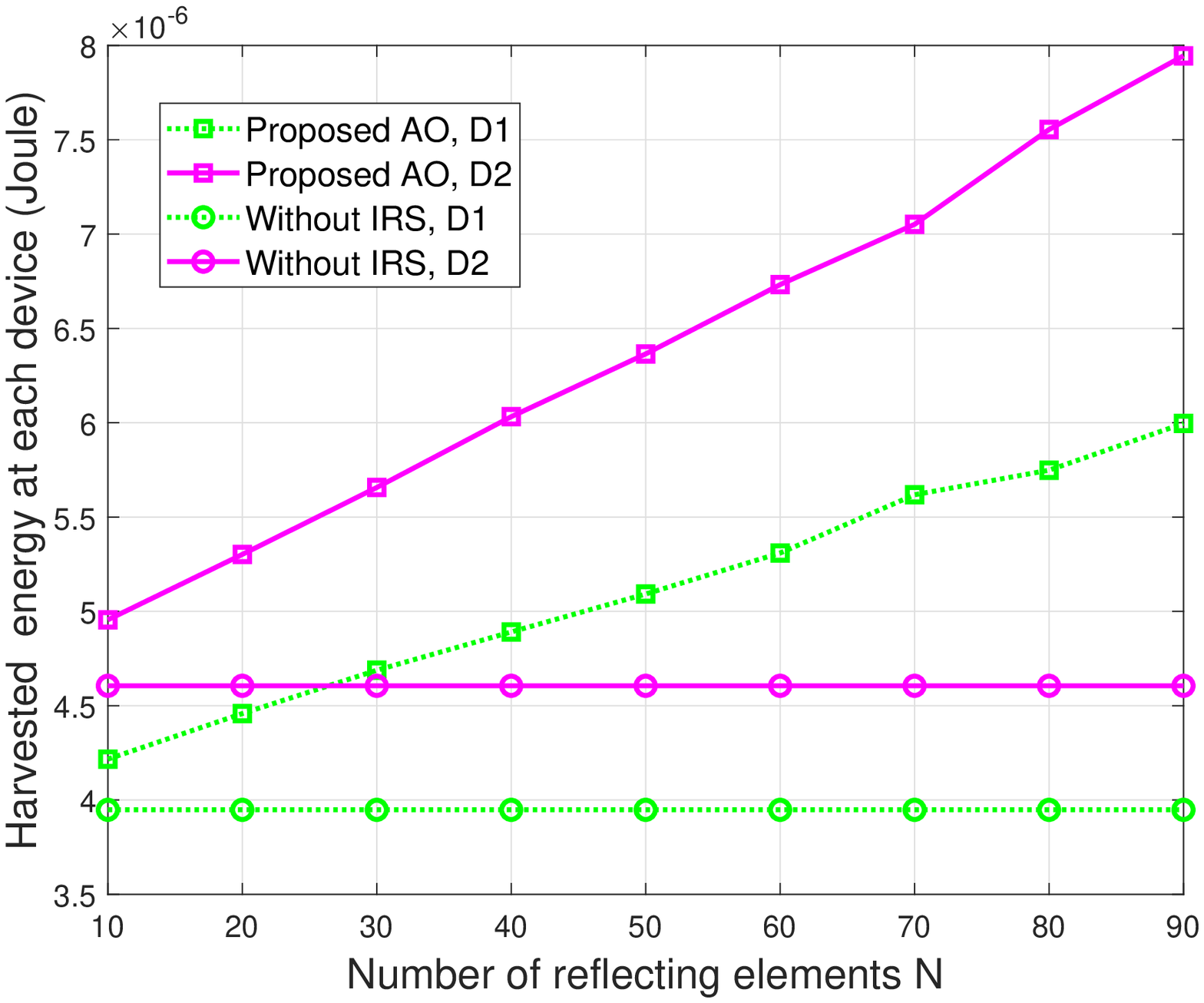}\label{N:versus:energy}}
\caption{Simulation results. } \label{pb}\vspace{-0.6cm}
\end{figure*}
\vspace{-0.1cm}
\subsection{Performance Comparison}
In Fig. \ref{N:versus:rate}, we plot the sum throughput versus the number of IRS elements for $K=5$ and $K=10$. {For comparison, we consider the following schemes: 1) ``Proposed JO with SDR'' where problem \eqref{probm:SDR:20} with ${\rm{rank}}(\bm{V}_0)=1$ in \eqref{P6:SDR:C10} relaxed  is solved successively, and thus, this scheme serves as a performance upper bound;  2) ``Proposed JO with GR'' where Gaussian randomization is applied to recover a rank-one $\bm{V}_0$ based on the solution of the scheme in 1); 3) ``Proposed JO with Penalty'' where we replace ${\rm{rank}}(\bm{V}_0)=1$  by ${\rm Tr}(\bm{V}_0) - \lambda_{\max}(\bm{V}_0) \leq  0$  with  $\lambda_{\max}(\bm{V}_0)$ denoting the largest singular value of $\bm{V}_0$ and then solve the resulting problem using the penalty method \cite{wu2019jointSWIPT};
4) ``Proposed AO"  in Section III-C;  5) Fixed time allocation with optimized IRS phase shifts, i.e.,  $\tau_0=0.5 T_{\max}$; 6) Fixed IRS phase shifts with optimized time allocation, i.e.,  $\theta_n=0, \forall n$; and 7) Without IRS.}

It is observed from Fig. \ref{N:versus:rate} that   the sum throughput gain achieved by our proposed JO/AO designs over the benchmark schemes increases as $N$ increases  for both $K=5$ and $K=10$. In particular, the proposed AO algorithm achieves almost the same performance as the proposed JO with SDR/GR/Penalty and  is thereby a practically appealing solution considering its low complexity, especially for large $N$.  Besides, the performance of the scheme with fixed phase shifts is less sensitive to increasing $N$ and only achieves a marginal gain over the system without IRS, whereas the scheme with fixed time allocation performs even  worse than the system without IRS for small $N$, but outperforms the scheme with fixed phase shifts for large $N$. This is expected since as $N$ increases, the passive beamforming gain achieved by phase-shift optimization helps compensate the performance loss incurred by fixed time allocation. Nevertheless, Fig. \ref{N:versus:rate} highlights the importance of the joint design of the IRS phase shifts and the time allocation.

\vspace{-0.3cm}
\subsection{Impact of IRS on WPCNs}
We note that the significant throughput improvement shown in Fig. \ref{N:versus:rate} is due to the deployment of the IRS and not due to an increase in the total energy consumption at the HAP, which is given by $E_{\rm HAP}=P_{\rm A}\tau_0$. To illustrate this explicitly, we plot in Fig.  \ref{N:versus:tau0} the optimized DL WPT duration $\tau_0$  versus $N$ obtained with the proposed AO algorithm and  without IRS.  It is observed that for IRS-assisted WPCNs, the optimized DL WPT duration  decreases as $N$ increases and thus the total system energy consumed at the HAP $E_{\rm HAP}$ is actually reduced. This also leaves devices more transmission time for UL WIT, which benefits the sum throughput since $R_{\rm sum}$ is monotonically increasing in $\tau_1$.  {This suggests that integrating IRS into WPCNs introduces a desirable ``double-gain'' as it not only improves the system throughput but also reduces the energy consumption, thus rendering this architecture spectrum and energy efficient.}

Furthermore,  although the DL WPT duration $\tau_0$ decreases due to the deployment of the IRS, it is not at the cost of reducing the energy harvested  at the devices. In Fig. \ref{N:versus:energy}, we plot the harvested energy of two randomly selected devices (e.g., D1 and D2), i.e., $E^h_k =\eta_kP_{\rm A}|h^H_{d,k} +   \bm{q}_k^H \vvv_0|^2\tau_0$, versus $N$ when $K=5$. One can observe that  the harvested energy monotonically increases with $N$ for both devices. Considering the decrease of $\tau_0$ shown in Fig.  \ref{N:versus:tau0}, it is not difficult to conclude that the increase of $E^h_k$ is solely due to the improved effective channel power gain $|h^H_{d,k} +   \bm{q}_k^H \vvv_0|^2$, which further demonstrates the effectiveness of IRS for WPCNs.
\vspace{-0.4cm}

\section{Conclusions}
This paper studied  IRS-assisted WPCNs employing NOMA for UL WIT.  
We first unveiled that   dynamic  IRS beamforming cannot improve the performance of the considered system, which simplifies the problem and reduces the signalling overhead. Based on this result, we proposed both joint and alternating optimization algorithms for system throughput maximization where the latter admits closed-form expressions and is practically appealing.  Numerical results showed that our proposed designs are able to  drastically improve the system performance compared to several baseline schemes and also shed light on how the IRS improves the throughput of WPCNs while decreasing the energy consumption as the number of IRS elements increases. {In future work, it is worth investigating the effectiveness of dynamic IRS beamforming for WPCNs with imperfect CSI and/or SIC, different user weights,  multiple antennas, etc.}


\vspace{-0.5cm}


\end{document}